\documentclass[letterpaper, 10 pt, conference]{ieeeconf}  

\IEEEoverridecommandlockouts                              

\usepackage{graphicx} 
\usepackage{amsmath} 
\usepackage{amssymb}  

\usepackage{amsfonts}
\usepackage{bbm}
\usepackage{subcaption}
\usepackage{ifthen}
\usepackage{cite}
\usepackage[usenames, dvipsnames]{color}
\usepackage{url}
\usepackage{hyperref}
\usepackage[hyphenbreaks]{breakurl}
\usepackage{mathtools}	
\usepackage{comment}
\usepackage[normalem]{ulem}

\usepackage{amsthm}
\usepackage{balance}

\newtheorem{theorem}{Theorem}
\newtheorem{lemma}{Lemma}

\newtheorem{corollary}{Corollary}

\newtheorem{problem}{Problem}
\newtheorem{remark}{Remark}
\newtheorem{assumption}{Assumption}

\newcommand\oprocendsymbol{\hbox{$\square$}}
\newcommand\oprocend{\relax\ifmmode\else\unskip\hfill\fi\oprocendsymbol}

\newcommand{\s}{\sqrt{\frac{\gamma}{a}}}

\renewcommand{\paragraph}[1]{
 \vspace{1ex}{\noindent\textbf{#1.}\\} 
}

\newcommand{\Z}{\frac{b \omega_0^2}{\eta_0(2\gamma\eta_0^2+\omega_0^2a)}}

\newcommand{\Y}{{\frac{\omega_0^2-\theta_0^2}{(\omega_0^2+\theta_0^2)^2}}}

\newcommand{\e}{{e^*}}

\newcommand{\red}[1]{{ \color{red} {{#1} }}}
\newcommand{\remove}[1]{{ }}

\title{On the Value of Energy Storage in Generation Cost Reduction}

\author{Yue Shen, 
Maxim Bichuch, and Enrique Mallada 
\thanks{
Y. Shen and E. Mallada are with the ECE Department of Electrical and Computer Engineering and M. Bichuch is with the Department of Applied Mathematics and Statistics, Johns Hopkins University, 3400 North Charles St., Baltimore, MD 21218, USA \texttt{\{yshen50,mbichuch,mallada\}@jhu.edu}.}
\thanks{
This work was supported by US DoE EERE award DE-EE0008006, NSF through grants  CNS 1544771, EPCN 1711188, AMPS 1736448, DMS 1736414, and CAREER 1752362,  and Johns Hopkins University Discovery Award.}
}

\begin{document}

\maketitle
\thispagestyle{empty}
\pagestyle{empty}

\maketitle
\begin{abstract}
This work seeks to quantify the benefits of using energy storage toward the reduction of the energy generation cost of a power system.  A two-fold optimization framework is provided where the first optimization problem seeks to find the optimal storage schedule that minimizes operational costs. Since the operational cost depends on the storage capacity, a second optimization problem is then formulated with the aim of finding the optimal storage capacity to be deployed. Although, in general, these problems are difficult to solve, we provide a lower bound on the cost savings for a  parametrized family of demand profiles.
The optimization framework is numerically illustrated using real-world demand data from ISO New England. Numerical results show that energy storage can reduce energy generation costs by at least 2.5\%.
\end{abstract}

\medskip


\section{Introduction}

The electric power grid is undergoing one of the most fundamental transformations since its inception~\cite{Liserre2010}.  Technological development of renewable energy sources~\cite{Omar2014} coupled with the need to reduce carbon emissions is transforming the generation mix~\cite{Solomon1704}.
Alongside, the electrification of transportation is driving a rapid growth on global electricity demand~ \cite{energyreport}. Among the many challenges that this paradigm shift introduces is the lack of synchronism between the times when renewable energy is available and the time when energy demand is required.
Energy storage is often seen as tentative solution towards addressing this challenge due to its ability to dispatch energy to shift energy availability across space, via deployment of distributed storage~\cite{Atzeni2013}, and across time~\cite{Zhao2015}.
Thus the additional flexibility that storage provides, together with the steady decrease on build and installation prices, has stimulated the deployment of several grid-scale storage systems, e.g., \cite{osti_1484345} and \cite{storagereport}.

However, despite the clear benefits that storage introduces, many questions regarding storage investment as well as efficient storage operation remained, to this day, unanswered. The key difficulty on this regard is that the cost of using storage is not a function of the instantaneous power, as it is the case for generators, where the cost can be mapped to the cost of the fuel consumed for generating electricity. Instead, the cost of using storage indirectly arises from the unit degradation that is experienced during charging/discharging cycles~\cite{Xu2018}.

This last observation is in contrast with the vast majority of existing works which formulate the cost of storage without detailed degradation model. Examples of this approach includes \cite{L2010, Q2011, Ozel2014},  which typically approach the storage operation by solving an optimization problem whose cost depend on fixed storage life-span. 
Only very recently, the question of how to optimally coordinate resources whose cost depend on instantaneous power with resources whose cost depend in energy trajectories, has started to be considered.%
In \cite{Mathias2019}, the problem of optimal coordination of limited-energy demand response and generation is considered. Similarly, \cite{Shi2019, Ju2018} consider the cost storage degradation and proposes an online algorithm with optimality guarantees.



In this paper, by perusing a path aligned to the works \cite{Shi2019,Ju2018}, we seek to quantify the economic benefit of using storage arbitrage as means to reduce energy production costs.
To this aim, we formulate a two-fold optimization framework aiming at: (a) finding the optimal storage operation that minimizes the total operational cost (including storage degradation cost); and (b) finding the optimal amount of storage that need to be deployed in the system to achieve the maximum benefit. Despite the complexity of such problems, we provide a sub-optimal policy for a paramterized family of demand trajectories that allow us quantify a lower bound on the potential  operational savings.


The rest of the paper is organized as follows. Our problem setting, including the energy system model, as well as a general problem formulation of the optimal \emph{operational problem} that seeks to optimally control storage to reduce operations costs, 
and the \emph{planning problem}, that seeks to find the optimal storage deployment, are presented in section \ref{sec:preliminaries}. In Section~\ref{sec:approach}, we describe the proposed approach to quantify the potential benefits by reformulating the above-mentioned general problems. The analytical solutions to our reformulations are presented in Section~\ref{sec:operational} and  Section~\ref{sec:planning}. Finally, preliminary numerical analysis, and conclusions, are provided in sections \ref{sec:numerical} and \ref{sec:Conclusion}, respectively.

\section{Problem Formulation}
\label{sec:preliminaries}
In this section, we describe the energy system that we seek to study and formalize the two problems that we consider  in this paper.

\subsection{System Model}
We use ${d}(t)\ge0$ to denote the net uncontrollable power demand --possibly including renewable-- and $p(t)$ the total aggregate power generation of a system operator at time $t\in [t_0,t_f]$. 
The total energy stored in the system at time $t$ is denoted by $e(t)$.
The energy  $e(t)$ evolves according to
\begin{equation}\label{eq:storage}
    \dot e(t) = u(t),
\end{equation}
where $u$ is the rate of change of stored energy.  We adopt the convention that $u(t) > 0$ implies charging, whereas $u(t) < 0$ means discharging. The total storage capacity is denoted by $C$ and the maximum charging/discharging rate by $r$, i.e., 
\begin{align}
    \label{eq:storage-constraints}
    -r\leq u(t)\leq r=\frac{C}{\epsilon},\qquad\text{and}\qquad 0\leq e(t)\leq C.
\end{align}
The ratio $\epsilon=\frac{C}{r}$ referred as technology parameter and it is aimed at representing different technological features of the storage. We further let $T_\text{ls}$ denote the total lifespan of the energy system.

Finally, the net load of the energy system ${d}(t) + u(t)$ is supplied by external power supply $p(t)$, i.e., 
\begin{equation}\label{eq:supp-dem-balance}
    {p}(t) = {d}(t) + u(t).
\end{equation} 
For simplicity, we assume in this paper that the charging/discharging process is lossless. A more realistic model that relaxes this assumption is subject of current research.

\subsection{Cost Model}
We are interested in quantifying the benefits of using energy storage as a way to reduce the overall  cost that the system incurs in meeting the demand $d(t)$. 

\subsubsection{Generation Cost}
We model the aggregate generation cost using  quadratic cost function $L_{g}: \mathbb{R}\rightarrow\mathbb{R}$ for power supply, i.e.,
\begin{align}\label{eq:gen-cost}
    &L_{g}( p) = \frac{a}{2} {p}^2+b {p},
\end{align}
where $a$ and $b$ are positive cost coefficients. Equation \eqref{eq:gen-cost} represents either the generation cost derived from fuel consumption (see, e.g.,~\cite{Adefarati2013}) or the integral of the inverse aggregate supply function derived empirically from the ISO (see, e.g.,~\cite{you2019role,tang2016model}).

\subsubsection{Storage Cost}
Unlike generation cost that originates from cost of producing energy, storage cost is a result of battery degradation that occurs with each discharge cycle. 
To compute this cost we first let functional $\mathcal D_i:C^{\infty}_{[t_0,t_f]}\rightarrow\mathbb{R}_{\ge0}$ denote the $i$th Depth of Discharge (DoD), $i\in\mathbb{N^+}$, of the stored energy trajectory $e\in C^{\infty}_{[t_0,t_f]}$. The DoD is the capacity that has been discharged from the fully charged battery. And it is often normalized as a fraction of the total capacity. Given the normalized DoD $y\in(0,1]$, the storage degradation is given by the cycle depth stress function $\Phi(y)$, with  $\Phi: \mathbb{R}_{\ge0}\rightarrow\mathbb{R}_{\ge0}$. This function quantifies the battery  loss of life due to a cycle of depth $y=\frac{D}{C}$ for a DoD $D$ and capacity $C$.

Thus the total storage cost due to the energy trajectory $e$ is represented  by the cost functional $\mathcal{L}_{s}:C^{\infty}_{[t_0,t_f]}\rightarrow\mathbb{R}_{\ge0}$ given by
\begin{align}
    &\mathcal{L}_{s}(e) = \sum_{i=1}^\infty \Phi\left(\frac{\mathcal D_i(e)}{C}\right) C \rho,\label{eqn:lse}
\end{align}
where $\rho$ represents the one-time unit building cost. 

\subsection{The Value of Storage}
As mentioned before, we seek to quantify the benefit that storage brings to an ISO. We will investigate this benefit in two settings. We first consider the \emph{operational problem} of how to optimally operate the available storage to minimize the system cost. Because such problem implicitly depends on the amount of storage available, we then move towards the \emph{planning problem} of finding the optimal amount of storage that one needs to deploy.


\medskip
\subsubsection{Operational Problem}
The first difficulty on seeking to optimally use storage arises from the difference in the argument between \eqref{eq:gen-cost} and \eqref{eqn:lse}.
While \eqref{eq:gen-cost} is a function of the instantaneous power being generated, \eqref{eqn:lse} depends on the entire energy trajectory, which implicitly depends on $u(t)$. 
We overcome this issue by expressing \eqref{eq:gen-cost} in terms of the total energy production cost over the interval $[t_0,t_f]$. This is given by the generation cost functional $\mathcal{L}_{g}:C^{\infty}_{[t_0,t_f]}\rightarrow\mathbb{R}_{\ge0}$, i.e.,
\begin{align}
    \mathcal{L}_{g}({p},{t_0},{t_f})&\!=\!\!\int_{t_0}^{t_f}\!\!\!L_g(p(t))dt 
    \!=\! \int_{t_0}^{t_f}\frac{a}{2} {p}(t)^2\!+\!b {p}(t)dt.\label{eqn:lg}
\end{align}

Using \eqref{eqn:lg} and \eqref{eqn:lse} one can formally define the optimal storage control problem as
\begin{subequations}\label{eqn:storage-control}
\begin{align}
    J(C,[t_0,t_f])=
    \min_{u} \quad&\mathcal{L}_{g}({p},{t_0},{t_f}) +\mathcal{L}_{s}(e,{t_0},{t_f}),\hfill\\
    \mathrm{s.t.}   \quad&\;\;\;
    \eqref{eq:storage},\eqref{eq:storage-constraints},\eqref{eq:supp-dem-balance}.\\[-2ex]\nonumber
\end{align}
\end{subequations}
The optimal control problem \eqref{eqn:storage-control} has the advantage of combining both, the cost of using the energy storage, together with the cost generation power in a common setting. 

However, there are two main difficulties in using \eqref{eqn:storage-control} towards quantifying the benefits of storage. Firstly, the solution will depend on the boundary conditions at $t_0$ and $t_f$, which can make the interpretation of the benefits hard to asses. Secondly, the cost functional \eqref{eqn:lse} is hard to evaluate and highly dependent on the demand $d$. 

We overcome the dependence on the boundary conditions by consider the following average operational cost problem. The latter issue will be addressed in the next section.
\begin{problem}(Average Operational Problem)\label{prob:average-op}
\begin{align}\label{eqn:operationalprob1}
    J(C)=
    \lim_{\substack{t_0\rightarrow-\infty\\t_f\rightarrow\infty}} \frac{1}{t_f \!-\ t_0}
    J(C,[t_0,t_f])
\end{align}
\end{problem}
Given the tuple $(C)$, Problem \ref{prob:average-op} quantifies the average operational cost $J(C)$. This value can be used to measure the benefits of including storage by looking at difference in cost between $J(0)$ and $J(C)$, i.e., $$B_f(C)= J(0)-J(C).$$

\medskip
\subsubsection{Planning Problem}
We now turn to the formulation of the planning problem.  More precisely, we seek to capture the effect of $C$ on the operational cost and find the optimal storage capacity. However, because certain values of $C$ may not be feasible, we will implicitly constraint them including a bound on the life-span $T_{ls}\leq T_{ls,\max}$. This leads to the following optimization problem.





\begin{problem}(Planning Problem)\label{prob:planning}
\begin{subequations}\label{eqn:planningproblem1}
\begin{align}
    \min_{C} \quad &J(C),\label{eqn:ppa}\\    
    \text{s.t.} \quad
                     &T^{u^*}_{ls}(C)\le T_{ls,\max}.\label{eqn:ppf}
\end{align}
where $T^u_{ls}(C)$ is the life span of a storage of capacity $C$ under the control $u$, and $u^*$ is the optimal policy derived in Problem \ref{prob:average-op}.
\end{subequations}
\end{problem}
We finalize by noting that the relationship between $T_{ls}$ and $(C)$ is not straightforward, and only appears in cases were the optimal capacity in \eqref{eqn:planningproblem1} leads to an optimal policy with very low storage degradation per unit of time that requires unrealistic life-spans.

\section{Solution Approach}
\label{sec:approach}
As mentioned before, problems \ref{prob:average-op} and \ref{prob:planning} are either intractable due to the complexity of evaluating $\mathcal L_s$ or uninformative of the overall benefits of using storage due to the dependence of \eqref{eqn:storage-control} on boundary conditions. In order to overcome this limitation, we relax some of the constraints and seek to find an upper bound on a family of instances of such problems. This allows us to characterize the dependence of $J(C)$ on the frequency and amplitude of intra-day demand cycles and, in this way, get an upper bound on the benefits that storage can introduce. 

We focus on demand functions that capture the fluctuating demand behavior. We start by assuming a realistic case where power demand $d$ is perturbed around certain baselines $d_0\ge0$. We further assume the perturbations around $d_0$ is a periodic sinusoidal deviation with amplitude $d_1 \le d_0$, i.e.
\begin{align}
     &{d}(t) = d_0 + d_1\sin(\omega_0t)\label{eqn:d}.
\end{align}

\subsection{Operational Problem Reformulation}
As mentioned above, instead of considering explicitly the storage functional $\mathcal L_s$ in \eqref{eqn:lse} that depends on the whole stored energy trajectory, we use
quadratic storage cost functional $\mathcal{L}_{q}(\cdot;
\gamma):C^{\infty}_{[t_0,t_f]}\rightarrow\mathbb{R}_{\ge0}$
\begin{equation}
    \mathcal L_q (e,{t_0},{t_f};\gamma) = \int_{t_0}^{t_f}  \frac{\gamma}{2}(e(t)-e_0)^2 dt
\end{equation}
that penalizes the instantaneous stored energy deviation from a reference energy $e_0$. 
The penalty parameter $\gamma > 0$ not only limits the amount of energy being used, thus limiting degradation, but it also implicitly constraints the control effort $u$. Thus we remove constraints \eqref{eq:storage-constraints} and solve instead the following auxiliary optimal control problem,
\begin{subequations}\label{eq:12}
\begin{align}
  \tilde J(\gamma,[t_0,t_f])= \min_{u}\,\, &\mathcal{L}_g(p,{t_0},{t_f}) + \mathcal{L}_q(e,{t_0},{t_f};\gamma)\label{eqn:aa},\\
    \mbox{s.t.} \,\,& p =  d + u,\label{eqn:ad}\\
                  &\dot e=u,\label{eqn:ae}
\end{align}
\end{subequations}
which after taking  $t_0\rightarrow-\infty$ and $t_f\rightarrow\infty$ leads to the following auxiliary operational problem.
\begin{problem}\label{prob:infinite}
(Auxiliary Average Operational Problem).
\begin{align}\label{eqn:appxproblem}
  \hat J(\gamma)=\lim_{\substack{t_0\rightarrow-\infty\\t_f\rightarrow\infty}} \frac{1}{t_f - t_0}  &\hat J(\gamma,[t_0,t_f]).
\end{align}
\quad \,where $ d$ is defined by \eqref{eqn:d}.
\end{problem}
\begin{remark}
For problem \ref{prob:infinite},  our approach is to firstly solve the finite time case and then extend it to infinite time horizon by taking limit of the starting and final time. By doing this, even though the solution to \eqref{eq:12} is optimal for any finite time interval, it raises questions about the sufficiency and uniqueness in infinite horizon setting. We will show that our constructive solution is optimal over the set of all bounded solutions, which is the set of solution elevates practical interest. We remark, however, that the optimal solution is no longer unique.
\end{remark}

Although the value $\hat J(\gamma)$ does not have a specific economic meaning, we will show that the optimal solution of \emph{Problem 3} is pure sinusoidal around certain baselines with frequency $\omega_0$. The amplitudes of $u(t)$ and $e(t)$ are functions of $\gamma$, i.e.,
\begin{subequations}
\label{eqn:tempsolution}
\begin{align}
    u(t,\gamma) &= -u_1(\gamma)\sin(\omega_0t),\\
    e(t,\gamma) &= e_0+e_1(\gamma) \cos(\omega_0t).
\end{align}
\end{subequations}

Note that for control and storage trajectories \eqref{eqn:tempsolution}, every cycle is identical. Thus, for this particular choice of demand and class of problems it is possible to express the DoD for every cycle as:
\begin{equation}
 D(\gamma) = 2 e_1(\gamma). \label{eqn:mde}
\end{equation}
as well as compute explicitly, the average generation cost
\begin{align}
    J_g(\gamma) = \frac{a}{4}(d_1-u_1(\gamma))^2+\frac{a}{2}d_0^2+b d_0,\label{eqn:jlg}
\end{align}
 the average storage cost of the original storage model $\mathcal{L}_s(e)$ 
\begin{align}
    J_{s}(\gamma, C) &= \Phi(2 e_1(\gamma)/C) C \rho \frac{\omega_0}{2\pi},\label{eqn:jls}
\end{align}
%
and the total operational cost of the control policy \eqref{eqn:tempsolution}:
\begin{align}
    J(\gamma,C)&= J_g(\gamma)+J_s(\gamma,C).\label{eqn:oper-cost-gamma}  
\end{align}
As a result, since \eqref{eqn:tempsolution} is a feasible solution to Problem \ref{prob:average-op}, it is an upper bound on the total operational cost.

\subsection{Planning Problem Reformulation}
Once we have explicit expressions for the long term average generation and storage cost as functions of $\gamma$ and $C$, we can further reformulate the planning problem in the following form. Note that we will optimize it with respect to $\gamma$ and $C$ and the solution of this reformulated planning problem gives optimal storage capacity and penalty parameter that achieves an upper bound on the optimal planning problem Problem \ref{prob:planning}.

\begin{problem}(Reformulated Planning Problem)
\begin{subequations}\label{eqn:planningproblem_r}
\begin{align}
    \min_{\gamma,C} \quad &J(\gamma,C),\label{eqn:ppa2}\\    
    \text{s.t.} \quad
                     &  D(\gamma)\leq C,\label{eqn:rpc}\\
                     & u_1(\gamma)\leq r=\frac{C}{\epsilon} ,\label{eqn:rpe}\\
                     &T^{u(t,\gamma)}_{ls}(C)\le T_{ls,\max}.\label{eqn:ppf2}
\end{align}
\end{subequations}
where ${D}(\gamma)$,  $J(\gamma,C)$, and $(u(t,\gamma),u_1(\gamma))$ are defined in \eqref{eqn:mde}, \eqref{eqn:oper-cost-gamma}, and \eqref{eqn:tempsolution}, respectively.
\end{problem}

\section{Operational Problem}
\label{sec:operational}
In this section we provide an analytical solution to the auxiliary problem as well as 
\emph{Theorem 1} unveils the analytically optimal solution of the auxiliary average operational problem. 

\begin{theorem}
\label{the:infinite}
Given the operational problem \eqref{eqn:appxproblem}, a optimal storage control over all bounded $e(t)$ in an infinite time horizon follows equation \eqref{eqn:tempsolution} with 
\begin{subequations}
\begin{align}
    u_1(\gamma) &= d_1\frac{\omega_0^2}{\theta^2+\omega_0^2},\\
    e_1(\gamma) &= d_1\frac{\omega_0}{\theta^2+\omega_0^2},\label{eq:e1}
\end{align}
\end{subequations}
where $\theta = \sqrt{\frac{\gamma}{a}}$.
\end{theorem}

\begin{proof}
By substituting \eqref{eqn:d} and \eqref{eqn:ad} into \eqref{eqn:aa} and defining $d_{s}(t):=d_1\sin(\omega_0 t)$ as the sinusoidal part of demand \eqref{eqn:d}, we can explicitly express the operational cost as
\begin{align}
    &\mathcal{L}({p},{e},{t_0},{t_f}; \gamma)=\mathcal{L}_g( p,{t_0},{t_f}) + \mathcal{L}_q( e,{t_0},{t_f};\gamma),\nonumber\\
    &\quad= \int_{t_0}^{t_f}\Big(\frac{a}{2} d_0^2 +a d_0( d_s (t)+u(t)) + \frac{a}{2} ( d_s(t)+u(t))^2\nonumber\\
    &\quad\quad+ b d_0 + b ( d_s(t)+u(t))+\frac{\gamma}{2} (e(t)-e_0)^2 \Big)dt.\label{eqn:operaational cost3}
\end{align}
Since the optimal control does not depend on the constant terms, we can drop the terms $bd_0$ and $\frac{a}{2}d_0^2$, and further let $p_s (t) : = d_s(t) + u(t)$ and $e_s(t): = e(t)-e_0$ to get
\begin{align}
    \mathcal{L}(p,e,{t_0},{t_f};\gamma)\! &=\!\!\int_{t_0}^{t_f}\!\frac{a}{2}  p_s(t)^2\!+\!\beta p_s(t)\!+\!\frac{\gamma}{2} e_s(t)^2dt ,\label{eqn:operational cost5}
\end{align}
where $\beta: =a d_0+b $.

By introducing a Lagrange multiplier $\lambda$ for the storage dynamics \eqref{eqn:ae}, the Euler-Lagrange equation 
demonstrates the optimal conditions for this problem:
\begin{subequations}
\begin{align}
    \label{eqn:p1tc1}
    a\left(d_s\left(t\right)+u\left(t\right)\right)+\beta+\lambda(t)=0,\\
    \label{eqn:p1tc2}
    \dot{\lambda}\left(t\right)=-\gamma e_s\left(t\right).
\end{align}
\end{subequations}

Combining \eqref{eqn:p1tc1} and \eqref{eqn:p1tc2} eliminates $\lambda$ and yields a second-order differential equation for which the optimal storage control needs to satisfy:
\begin{align}
    \label{eqn:p1tc}
    a\left(\dot{d}_s\left(t\right)+\dot{u}\left(t\right)\right)=\gamma e_s(t).
\end{align}

At the starting time $t_0$, let $d^s_{t_0}:=d_s(t_0)$, $e^s_{t_0}:=e_s(t_0)$ and $\dot{e}^s_{t_0}:= \dot e_s(t_0)$. Then the closed-form solution to \eqref{eqn:p1tc} is 
\begin{subequations}
\label{eqn:gsolution}
\begin{align}
    \label{eqn:gsolution1}
    e_s\left(t\right)&=\!\!\left(\!\!\!-d_1\!\frac{\theta}{\theta^2+\omega_0^2}\!\sin\!\left(\omega_0t_0\right)\!+\!\frac{1}{\theta}\!\left({\dot{e}^s}_{t_0}\!\!+\! d^s_{t_0}\right)\!\!\right)\!\sinh{\left(\!\theta\!\left(t\!-\!t_0\right)\!\right)}\nonumber\\
                   &+\left(-d_1\frac{\omega_0}{\theta^2+\omega_0^2}\cos\left(\omega_0t_0\right)+e^s_{t_0}\right)\cosh{\left(\theta\left(t-t_0\right)\right)}\nonumber\\
                   &+d_1\frac{\omega_0}{\theta^2+\omega_0^2}\cos (\omega_0t),\\
    \label{eqn:gsolution2}
    u\left(t\right)&=\!\!\left(\!\!\!-d_1\frac{\theta^2}{\theta^2+\omega_0^2}\!\sin\left(\omega_0t_0\right)\!+\!\left({\dot{e}^s}_{t_0}\!+\! d^s_{t_0}\right)\!\!\right)\!\cosh{\!\left(\!\theta\!\left(t\!-\!t_0\right)\!\right)\!}\nonumber\\
         &+\theta\left(-d_1\frac{\omega_0}{\theta^2+\omega_0^2}\cos\left(\omega_0t_0\right)+e^s_{t_0}\right)\sinh{\left(\theta\left(t-t_0\right)\right)}\nonumber\\
                   &-d_1\frac{\omega_0^2}{\theta^2+\omega_0^2}\sin (\omega_0t),
\end{align}
\end{subequations}
where $\theta = \sqrt{\frac{\gamma}{a}}$, $\dot{e}^{s}_{t_0}$ and ${e}^{s}_{t_0}$ are unknowns.

We further check the sufficient condition for optimally by  solving the Jacobi Accessory Equation \eqref{eqn:jaa} to \eqref{eqn:jac} with conditions in \eqref{eqn:jad}\cite[p.~52]{Liberzon2011}
\begin{subequations}
\begin{align}
    &L(u,e_s) = \frac{a}{2}(d_s+u)^2+\beta (d_s+u)+\frac{\gamma}{2} (e_s)^2;\\
    &R \!=\! \frac{1}{2}\left( \frac{\partial L}{\partial^2u}{}({u},{e_s})\right)\Bigg\vert_{\substack{d_s = d_s(t)\\ e_s=e_s(t)\\u=u(t)}}=\frac{1}{2}a;\label{eqn:jaa}\\
    &Q \!= \! \frac{1}{2}\!\left(\!\frac{\partial L}{\partial^2{e_s}}({u},{e_s})\!-\! \frac{d}{dt}\frac{\partial L}{\partial u\partial{e_s}}({u},{e_s})\!\right)\Bigg\vert_{\substack{d_s = d_s(t)\\ e_s=e_s(t)\\u=u(t)}}\nonumber\\
    &\quad=\frac{1}{2}\gamma;\label{eqn:jab}\\
    &Qv(t) = \frac{d}{dt}(R\dot{v}(t));\label{eqn:jac}\\
    &v(t_0) = 0;\quad v(t) \not\equiv 0.\label{eqn:jad}
\end{align}
\end{subequations}

The solution is the given by
\begin{align}
    &v(t) = c_1 e ^{\sqrt{\frac{\gamma}{a}}(t-t_0)}-c_1.
\end{align}
for some constant $c_1$. 
Note that the solution $v(t)\neq 0$ for all $t>t_0$ leads to non-existence of conjugate points within the time interval $[t_0,t_f]$. This, together with the fact that $R>0$, shows that the extremal is a strict minimum.

Now, we are able to derive the closed-form expressions of $e_s(t)$ and $u_s(t)$ with two remaining unknowns $e^s_{t_0}$ and $\dot{e}^s_{t_0}$.
Since $e_s(t_0)$ and $e_s(t_f)$ are constraint-free, by the transversality condition\cite[p.~83]{Liberzon2011} $\lambda(t_0)=\lambda(t_f)=0$, it follows from \eqref{eqn:p1tc1} that
\begin{subequations}
\label{eqn:trans}
\begin{align}
&\dot{e}^{s}_{t_0} = u(t_0) = -\frac{\beta}{\alpha}-d_{t_0}^s,\label{eqn:trans1}\\
&u(t_f) = -\frac{\beta}{\alpha}-d(t_f). \label{eqn:trans2}
\end{align}
\end{subequations}

Note that \eqref{eqn:trans1} immediately gives rise to the $\dot{e}^{s}_{t_0}$. Then evaluating \eqref{eqn:gsolution2} at $t_f$ and plugging in the expression (30b) yields an algebraic equation in $e^s_{t_0}$

Substituting both unknowns in \eqref{eqn:gsolution} gives, 

\begin{subequations}
\label{eqn:finsolution}
\begin{align}
    e_s\left(t\right)&=\!\sinh{\left(\theta\!\left(t-t_0\right)\!\right)}\left(\!\!-d_1\!\frac{\theta}{\theta^2+\omega_0^2}\!\sin\!\left(\omega_0t_0\right)\!-\!\frac{1}{\theta}\!\frac{\beta}{a}\right)\nonumber\\
    &+\frac{\cosh{\left(\theta\left(t-t_0\right)\right)}}{ \sinh(\theta T_{\Delta})}\left(\frac{1}{\theta}\frac{\beta}{a}(\cosh(\theta T_{\Delta})-1)\right.\nonumber\\
    &\left.+d_1\frac{\theta}{\theta^2+\omega_0^2}(\sin(\omega_0 t_0)\cosh(\theta T_{\Delta})-\sin(\omega_0 t_f))\right)\nonumber\\
    &+d_1\frac{\omega_0}{\theta^2+\omega_0^2}\cos (\omega_0t);\label{eqn:finsolution1}\\
    u\left(t\right)&=\!\cosh{\left(\theta\!\left(t-t_0\right)\!\right)}\left(\!\!-d_1\!\frac{\theta^2}{\theta^2+\omega_0^2}\!\sin\!\left(\omega_0t_0\right)\!-\!\frac{\beta}{a}\right)\nonumber\\
    &+\frac{\sinh{\left(\theta\left(t-t_0\right)\right)}}{ \sinh(\theta T_{\Delta})}\left(\frac{\beta}{a}(\cosh(\theta T_{\Delta})-1)\right.\nonumber\\
    &\left.+d_1\frac{\theta^2}{\theta^2+\omega_0^2}\left(\sin(\omega_0 t_0)\cosh(\theta T_{\Delta})-\sin(\omega_0 t_f)\right)\right)\nonumber\\
                   &-d_1\frac{\omega_0^2}{\theta^2+\omega_0^2}\sin (\omega_0t);\label{eqn:finsolution2}
\end{align}
where $T_{\Delta}:=t_f-t_0$ is the total operational time.
\end{subequations}

After taking the limits $t_f\rightarrow\infty$ and $t_0\rightarrow-\infty$, the non-periodical terms in \eqref{eqn:finsolution} that rely on the initial and final conditions vanish. Indeed, we have:
\begin{subequations}
\label{eqn:infsolution}
\begin{align}
    \lim_{\substack{t_0\rightarrow-\infty\\t_f\rightarrow\infty}}   e_s(t) &= d_1\frac{\omega_0}{\theta^2+\omega_0^2}\cos (\omega_0t);\\
    \lim_{\substack{t_0\rightarrow-\infty\\t_f\rightarrow\infty}}   u(t) &= -d_1\frac{\omega_0^2}{\theta^2+\omega_0^2}\sin (\omega_0t);\\
    \lim_{\substack{t_0\rightarrow-\infty\\t_f\rightarrow\infty}}   \lambda(t) &=-\alpha d_1\frac{\theta^2}{\theta^2+\omega_0^2}\sin (\omega_0t)-\beta.
\end{align}
\end{subequations}

Finally, we observe that one can show the expression $\frac{1}{T_\Delta}\lambda(t)(\hat{e}(t)-e_s(t))$ converges to 0 as $t_0\rightarrow-\infty$ and $t_f \rightarrow\infty$ for all $\hat{e}(t)$ that have sub-linear growth in $T_\Delta$.

 By \cite{CARTIGNY20031007}, it is sufficiently to state the limiting policy \eqref{eqn:infsolution} is an optimal solution of \eqref{eqn:appxproblem}.
\end{proof}

\begin{corollary}
Given the operational problem \eqref{eqn:appxproblem}, the optimal average generation cost $\mathcal{L}_s(e)$ and average storage cost  $\mathcal{L}_s(e)$ are expressed in \eqref{eqn:jlg} and \eqref{eqn:jls}, respectively.
\end{corollary}

\begin{proof}
We find the average generation cost by substituting \eqref{eqn:finsolution2} into the generation cost model $\mathcal{L}_g(p)$ defined in \eqref{eqn:lg} and then taking the average, i.e.
\begin{align}
    J_g(\gamma) &= \lim_{\substack{t_0\rightarrow-\infty\\t_f\rightarrow\infty}}\frac{1}{t_f-t_0}\mathcal{L}_g(p,{t_0},{t_f})\nonumber\\
                  &=\frac{a}{4}(d_1-u_1(\gamma))^2+\frac{a}{2}d_0+b d_0.
\end{align}

Moreover, since we found the explicit expressions of $\mathcal{D}(e)$ and $r(\gamma)$ for every cycle as functions of $\gamma$, we can express the long term average storage cost by the original storage model $\mathcal{L}_s(e)$ that we defined in \eqref{eqn:lse}, i.e.

\begin{align}
    J_{s}(\gamma,C) &= \lim_{\substack{t_0\rightarrow-\infty\\t_f\rightarrow\infty}}\frac{1}{t_f - t_0}\sum_{i=1}^\infty \Phi(2 e_1(\gamma)/C) C \rho\nonumber\\
    &=\Phi(2 e_1(\gamma)/C) C \rho \frac{\omega_0}{2\pi}.\label{eqn:js}
\end{align}
Thus, result follows.
\end{proof}
We can further use Theorem 1 to explicitly compute the life span of the storage under the control policy of \eqref{eqn:tempsolution}. 

\begin{corollary}\label{co:life-span}
Given the operational problem \eqref{eqn:appxproblem}, the life span of a storage with capacity $C$ under the optimal control $u(t,\gamma)$ in \emph{Theorem 1} is
\begin{align}
    T^{u(t,\gamma)}_{ls}(C)=\Phi(2e_1(\gamma)/C)^{-1}\frac{2\pi}{\omega_0}.
\end{align}
\end{corollary}

\section{Planning Problem}
\label{sec:planning}

We now leverage Theorem \ref{the:infinite} and Corollary \ref{co:life-span} to solve our operational problem \eqref{eqn:planningproblem_r}. 
To do this we first require a model on the storage degradation.
Thus, we further assume the following storage degradation cost function experimentally derived in \cite{Xu2018}.
\begin{assumption}\label{as1}
For a battery with monotonically increasing cycle depth stress function $\Phi (y)$, we assume $\Phi(y)y^{-1}$ is strongly convex, which holds if and only if
\begin{align}
    (\Phi(y)y^{-1})''>0.
\end{align}
\end{assumption}
\begin{lemma}
\label{le:dagradation}
Given the planning problem \eqref{eqn:ppa2}-\eqref{eqn:ppf2} with degradation model in \emph{Assumption \ref{as1}}, we can express the optimal cycle depth $y^*$ and penalty parameter $\gamma^*$ as follows

\begin{subequations}
\begin{align}
    y^* &=[y_s]^{\min{(y_{ub},1)}}_{y_{lb}}\\
\gamma^* &=\left\{
    \begin{aligned}
         \infty, &\qquad\text{if } \gamma_s<0,\\
         \gamma_s,&\qquad\text{o.w.}.
    \end{aligned}\right.\hfill
\end{align}
where $$\gamma_s := {2\omega_0^2a\rho}/\left({{d_1}\pi a \Phi(y^*)^{-1}y^*-2\rho}\right)$$ and $y_s$ solves $(\Phi(y)y^{-1})'=0$. $(\gamma_s, y_s)$ is the stationary point. And $y^*$ is the projection of $y_s$ onto the interval of $[y_{lb}, \min (y_{ub}, 1)]$. In addition, $y_{lb}:= \Phi^{-1}((T_{ls,\max}\frac{\omega_0}{2\pi})^{-1})\in(0,1]$ and $y_{ub} := \frac{2}{\epsilon\omega_0}\ge y_{lb}$ are lower and upper bounds on $y$ introduced by the maximum life span constraint \eqref{eqn:ppf2} and maximum charging/discharging rate constraint \eqref{eqn:rpe}, respectively.
\end{subequations}
\end{lemma}

\begin{proof} 
Recalling $y=D/C= 2 e_1(\gamma)/C$ is the normalized DoD, we can rewrite \eqref{eqn:jls} such that the cost function \eqref{eqn:ppa2} is
\begin{align}
\tilde J(\gamma,y) &\!=\! J(\gamma,C)
\!=\!J_g(\gamma)+2e_1(\gamma) \Phi(y)  y^{-1} \rho \frac{\omega_0}{2\pi},
\end{align}
with constraints on $C$ (19b) to (19d) and on $\gamma$ becoming:
\begin{subequations}
\label{eqn:c}
\begin{align}
        \max(0,\Phi^{-1}((T_{ls,\max}\frac{\omega_0}{2\pi})^{-1})]&\le y\le \min[1,\frac{2}{\epsilon\omega_0}],\label{eqn:c0}\\
        \gamma&\ge0.
\end{align}
\end{subequations}
Clearly, for a solution to the problem to exist, we require the above equations to have a non-empty intersection, i.e.,
\begin{align}
     \Phi^{-1}((T_{ls,\max}\frac{\omega_0}{2\pi})^{-1})&\le\frac{2}{\epsilon\omega_0}. \label{eqn:c2b}
\end{align}

Under this assumption, we can solve this optimization problem with respect to $y$ and $\gamma$. 
Solving KKT conditions gives unique KKT point $(\gamma^*, y^*)$ expressed in Lemma 1. We then claim that $(\gamma^*, y^*)$ is the optimal cycle depth and penalty parameter for all $y \in (0,1]$ and $\gamma\in[0,\infty)$. Indeed, we calculate the bordered Hessian:
\begin{align}
    \tilde H^{b} = \begin{bmatrix}
         0 & \frac{\partial}{\partial \gamma}\tilde J&\frac{\partial}{\partial y}\tilde J\\
         \frac{\partial}{\partial \gamma}\tilde J & \frac{\partial^2}{\partial \gamma \partial \gamma}\tilde J&\frac{\partial^2}{\partial \gamma\partial y}\tilde J\\
         \frac{\partial}{\partial y}\tilde J & \frac{\partial^2}{\partial \gamma\partial y}\tilde J&\frac{\partial^2}{\partial y\partial y}\tilde J
         \end{bmatrix}.
\end{align}

Note that both the first order leading principal minor and determinate of $\tilde H^{b}$ are strictly negative under \emph{Assumption \ref{as1}}, which is sufficient to show $\tilde J(\gamma, y)$ is quasi-convex\cite{Kenneth1961}. 

When the constraints \eqref{eqn:c} are not binding, i.e. $(\gamma^*, y^*) = (\gamma_s, y_s)$, we have the Hessian of Lagrangian evaluated at $(\gamma^*, y^*)$ Positive-definite. Thus the $(\gamma^*, y^*)$ is the strict local minima point. 

Also, for the case that the KKT point binding any constraint, we have $\nabla\tilde J(\gamma^*, y^*)\neq0$ and  $\tilde J(\gamma^*, y^*)$ twice differentiable in the neighborhood of $(\gamma^*, y^*)$.

In both cases, the unique KKT point globally minimize this quasi-convex objective function\cite{Kenneth1961}. Then \emph{Lemma \ref{le:dagradation}} follows.    
\end{proof}

To understand the implications of Lemma \ref{le:dagradation}, it is useful to consider the explicit dependence of $\tilde J(\gamma,y)$ with respect to $u_1 = \omega_0 e_1$ and $p_1:=d_1-u_1$ using equations:
\begin{subequations}
\begin{align}
\tilde J(p_1,u_1)&=\tilde J_{s}(u_1)+\tilde J_{g}(p_1);\\
\tilde J_{s}(u_1) &= u_1 \Phi(y)  y^{-1} \rho \frac{1}{\pi};\\
\tilde J_{g}(p_1) &= \frac{a}{4}p_1^2+\frac{a}{2}d_0+b d_0.
\end{align}
\end{subequations}

By taking the derivative of the storage cost $\tilde J_{s}(u_1)$ and the generation cost $\tilde J_{g}(p_1)$, we can get marginal cost for both cases:
\begin{subequations}
\begin{align}
    \frac{\partial}{\partial u_1} \tilde J_{s}(u_1) =& \Phi(y)y^{-1}\rho\frac{1}{\pi};\\
    \frac{\partial}{\partial p_1} \tilde J_{g}(p_1) =&\frac{\alpha}{2}p_1.
\end{align}
\end{subequations}

With these marginal costs, one can observe that it is not optimal to use any storage under the following condition, 
\begin{align}
    \frac{\partial}{\partial u_1} \tilde J_{s}(u_1)>\frac{\partial}{\partial p_1}\tilde J_{g}(p_1),\quad\forall u_1+p_1 = d_1.\label{eqn:nostoragecondition}
\end{align}

More precisely, $p_1 = d_1$ and $u_1 = 0$ solve the following problem when \eqref{eqn:nostoragecondition} holds.
\begin{subequations}
\begin{align}
    \min_{p_1,u_1} \quad &\tilde J(p_1,u_1)\\    
    \text{s.t.} \quad&d_1 = p_1+u_1.
\end{align}
\end{subequations}
\begin{remark}
When $\gamma_s<0$, Lemma \ref{le:dagradation} states that $\gamma^*=\infty$, which implies that the optimal amount of storage is \emph{zero}. This can occur, in particular, when ${d_1}\pi a \Phi(y^*)^{-1}y^*<2\rho$.
Note that this inequality can be written as 
\begin{align}
    \frac{\partial}{\partial u_1} \tilde J_{s}(u_1)>\frac{\partial}{\partial p_1}\tilde J_{g}(p_1)\Big\vert_{p_1 = d_1},
\end{align}
which implies condition \eqref{eqn:nostoragecondition}.

\end{remark}

\begin{theorem}
\label{the:optimal}
Given the planning problem \eqref{eqn:ppa2}-\eqref{eqn:ppf2} with degradation model in \emph{Assumption \ref{as1}}, the optimal capacity is
\begin{align}
    C^* &=2e_1(\gamma^*)/y^*,\hfill
\end{align}
where $(\gamma^*,y^*)$ are defined in Lemma \ref{le:dagradation} and $e_1(\gamma)$ is defined in \eqref{eq:e1}.
\end{theorem}

\begin{proof}
For optimization problem \eqref{eqn:ppa2}-\eqref{eqn:ppf2}, $(\gamma^*, C^*)$ satisfy all the KKT conditions and $J(\gamma,C)$ is also quasi-convex. Similarly, when the constraints (19b) to (19d) and $\gamma>0$ are not binding, we have the Hessian of Lagrangian evaluated at $(\gamma^*, C^*)$ positive-definite. And for the case that the KKT point binding any constraint, we have $\nabla J(\gamma^*, C^*)\neq0$ and  $J(\gamma^*, C^*)$ twice differentiable in the neighborhood of $(\gamma^*, C^*)$. 

Therefore, by \cite{Kenneth1961} again, $(\gamma^*, C^*)$ is the global minima point.
\end{proof}

\section{Numerical Results}
\label{sec:numerical}
In this section, we use our analysis to shed light on the potential savings that can be achieved, by  using real-world demand data from ISO New England \cite{iso} on (data: 07/17/2019). Based on the mentioned data, we choose $d_0$ and $d_1$ in \eqref{eqn:d} by selecting the DC and first harmonic of a one-day data. Thus,  $d_0 = 18091\, \mbox{MW}$,  $d_1 = 4671\, \mbox{MW}$, and $\omega_0  = 0.26 \, rad/hr$, which corresponds to one cycle per day. Therefore the demand variation $d_1\cos(\omega_0t)$ in \eqref{eqn:d} represents the demand fluctuations for the main daily cycle.

The storage degradation function in our numerical test is:

\begin{align}
    \Phi(y)=(k_1(y)^{k_2}+k_3)^{-1},\label{eqn:a1}
\end{align}
where $k_1 = 1.4\times10^5$, $k_2 = -0.5$, $k_3 = -1.23\times10^5$ and the largest battery life spans $T_{max} = 76$ years according to the LMO battery degradation test \cite{Xu2018}. Note that this function satisfy our general assumptions in Assumption \ref{as1}.

The generation cost coefficients are selected to be $a = 0.02$ and $b = 16.24$ base on the energy system model \cite{Adefarati2013}. And the capacity-power ratio is set as $\epsilon = 2$ by \cite{tesla}. At last, by \cite{Xu2018}, the one-time unit building cost $\rho = 209000\mbox{USD/MWh}$.

\begin{figure}[t!]
    \centering
    \begin{subfigure}[b]{0.45\textwidth}
        \includegraphics[width=\textwidth]{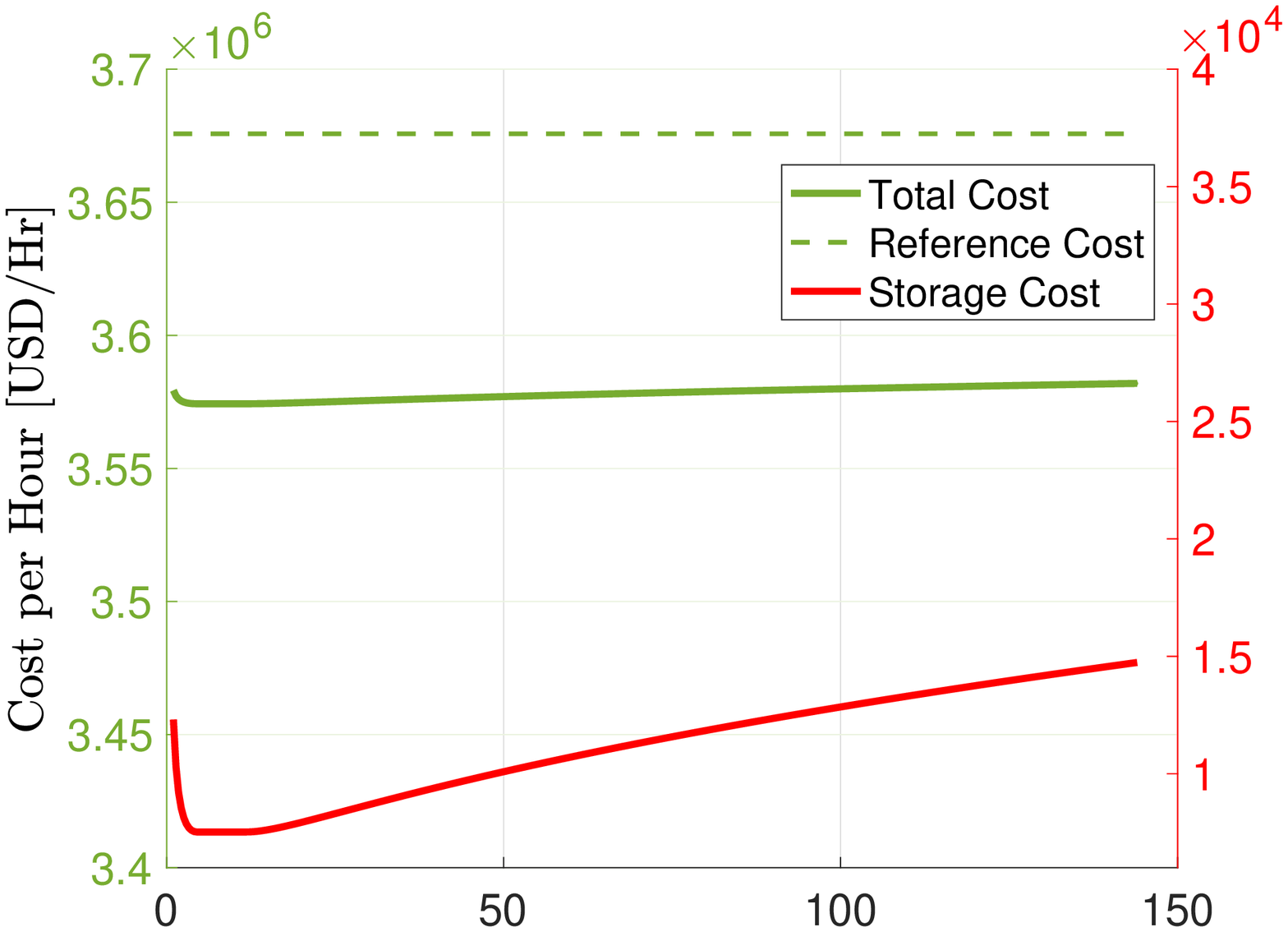}
    \end{subfigure}
    \quad
    \begin{subfigure}[b]{0.45\textwidth}
        \includegraphics[width=\textwidth]{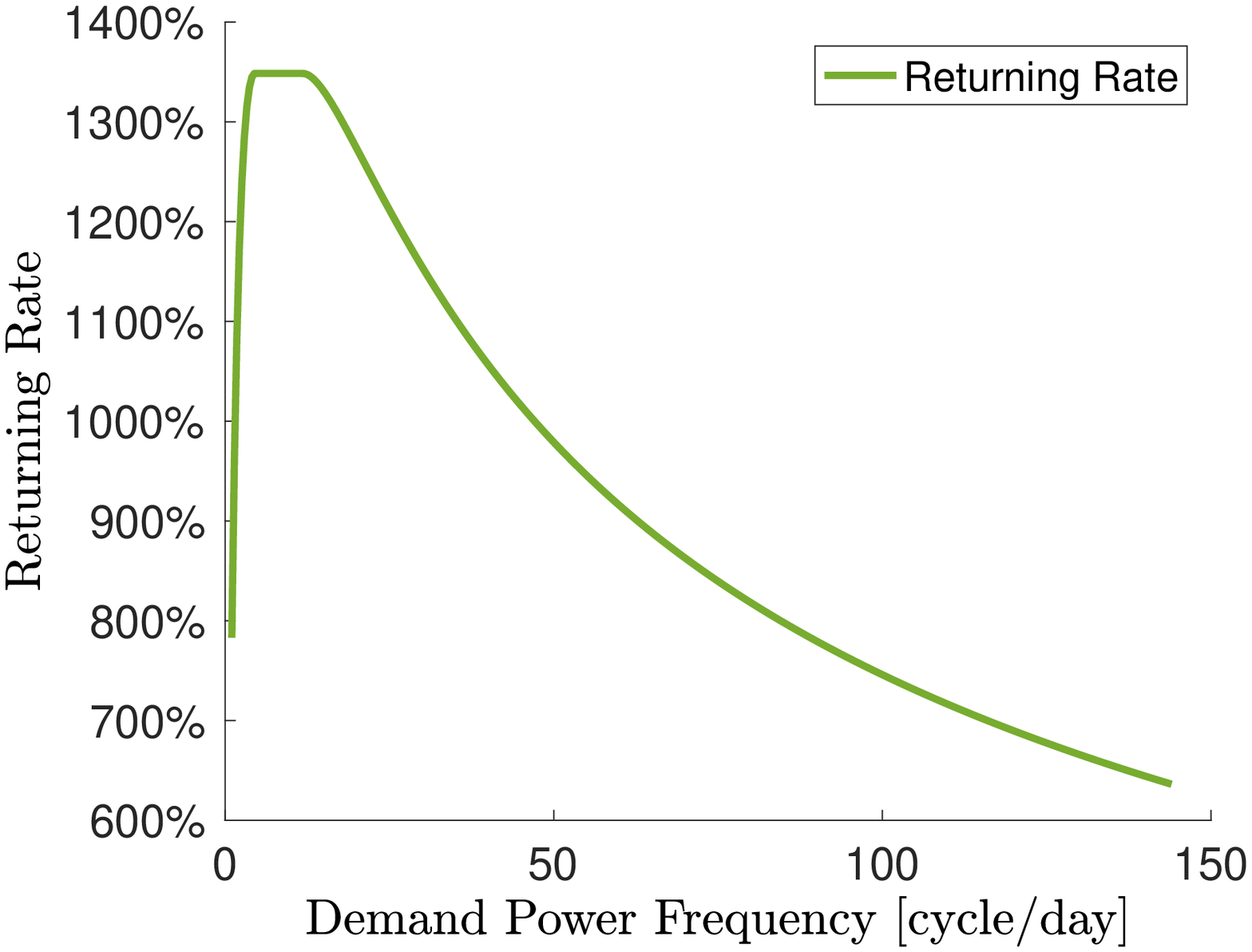}
    \end{subfigure}
    \caption{Cost composition (upper) and returning rate (lower). versus demand power frequency}
    \label{fig:w02}
\end{figure}

We study the impact of varying the demand frequency $\omega_0$ on the cost performance by varying $\omega_0$ from 1 cycle per day ($\omega_0 = 0.26 \, rad/hr$) to 144 cycle per day ($\omega_0 = 37.66 \, rad/hr$). Note that the non-empty intersection assumption \eqref{eqn:c2b} is satisfied for the given frequency range. 
The cost composition are plotted in the upper panel of Fig ~\ref{fig:w02} in where the reference cost is the cost without using any storage. Moreover, we calculate the ratio between saving and investment on storage as returning rate  and plot it in the lower panel of Fig ~\ref{fig:w02}.

Simulation results indicate cost saving between $2.77\% - 2.56\%$ can be achieved. Moreover, the storage is not used for both extremely low and extremely high frequencies (beyond the range of the plot). This is intuitive, since the storage operation will only incur extra cost when the demand is flat. Similarly, if the demand variation are extremely frequent, the storage runs out of degrades quickly and the building cost increases. Our results also show that larger generation cost coefficient $a$ leads to higher savings.
\section{Conclusion}
\label{sec:Conclusion}
In this work, we propose an optimization framework that aims at estimating the operational  cost benefits  of  using storage in an energy system as well as the optimal storage amount that should be deployed into the system.  Analytical closed form solutions to this framework bring insight on the effect of the frequency of demand fluctuations and capacity on the overall system cost. Numerical examples are provided to illustrate our framework. Analysing the effect of charging/discharging inefficiency on this framework is an ongoing topic of research.

\balance

\bibliographystyle{IEEEtran}
\bibliography{main.bib}

\begin{thebibliography}{10}
\providecommand{\url}[1]{#1}
\csname url@samestyle\endcsname
\providecommand{\newblock}{\relax}
\providecommand{\bibinfo}[2]{#2}
\providecommand{\BIBentrySTDinterwordspacing}{\spaceskip=0pt\relax}
\providecommand{\BIBentryALTinterwordstretchfactor}{4}
\providecommand{\BIBentryALTinterwordspacing}{\spaceskip=\fontdimen2\font plus
\BIBentryALTinterwordstretchfactor\fontdimen3\font minus
  \fontdimen4\font\relax}
\providecommand{\BIBforeignlanguage}[2]{{%
\expandafter\ifx\csname l@#1\endcsname\relax
\typeout{** WARNING: IEEEtran.bst: No hyphenation pattern has been}%
\typeout{** loaded for the language `#1'. Using the pattern for}%
\typeout{** the default language instead.}%
\else
\language=\csname l@#1\endcsname
\fi
#2}}
\providecommand{\BIBdecl}{\relax}
\BIBdecl

\bibitem{Liserre2010}
M.~{Liserre}, T.~{Sauter}, and J.~Y. {Hung}, ``Future energy systems:
  Integrating renewable energy sources into the smart power grid through
  industrial electronics,'' \emph{IEEE Industrial Electronics Magazine},
  vol.~4, no.~1, pp. 18--37, Mar. 2010.

\bibitem{Omar2014}
O.~Ellabban, H.~Abu-Rub, and F.~Blaabjerg, ``Renewable energy resources:
  Current status, future prospects and their enabling technology,''
  \emph{Renewable and Sustainable Energy Reviews}, vol.~39, pp. 748 -- 764,
  Nov. 2014.

\bibitem{Solomon1704}
S.~Solomon, G.-K. Plattner, R.~Knutti, and P.~Friedlingstein, ``Irreversible
  climate change due to carbon dioxide emissions,'' \emph{Proceedings of the
  National Academy of Sciences}, vol. 106, no.~6, pp. 1704--1709, 2009.

\bibitem{energyreport}
\BIBentryALTinterwordspacing
``{Global Energy \& CO2 Status Report},'' 2018. [Online]. Available:
  \url{https://www.iea.org/geco/}
\BIBentrySTDinterwordspacing

\bibitem{Atzeni2013}
I.~{Atzeni}, L.~G. {Ordóñez}, G.~{Scutari}, D.~P. {Palomar}, and J.~R.
  {Fonollosa}, ``Demand-side management via distributed energy generation and
  storage optimization,'' \emph{IEEE Transactions on Smart Grid}, vol.~4,
  no.~2, pp. 866--876, Jun. 2013.

\bibitem{Zhao2015}
H.~Zhao, Q.~Wu, S.~Hu, H.~Xu, and C.~N. Rasmussen, ``Review of energy storage
  system for wind power integration support,'' \emph{Applied Energy}, vol. 137,
  pp. 545 -- 553, Jan. 2015.

\bibitem{osti_1484345}
R.~Fu, T.~W. Remo, and R.~M. Margolis, ``2018 u.s. utility-scale
  photovoltaics-plus-energy storage system costs benchmark,'' \emph{2018 U.S.
  Utility-Scale Photovoltaics-Plus-Energy Storage System Costs Benchmark}, Nov.
  2018.

\bibitem{storagereport}
\BIBentryALTinterwordspacing
``{U.S. Battery Storage Market Trends}.'' [Online]. Available:
  \url{https://www.eia.gov/analysis/studies/electricity/batterystorage/pdf/battery_storage.pdf}
\BIBentrySTDinterwordspacing

\bibitem{Xu2018}
B.~{Xu}, A.~{Oudalov}, A.~{Ulbig}, G.~{Andersson}, and D.~S. {Kirschen},
  ``Modeling of lithium-ion battery degradation for cell life assessment,''
  \emph{IEEE Transactions on Smart Grid}, vol.~9, no.~2, pp. 1131--1140, Mar.
  2018.

\bibitem{L2010}
L.~{Xiaoping}, D.~{Ming}, H.~{Jianghong}, H.~{Pingping}, and P.~{Yali},
  ``Dynamic economic dispatch for microgrids including battery energy
  storage,'' in \emph{The 2nd International Symposium on Power Electronics for
  Distributed Generation Systems}, Jun. 2010, pp. 914--917.

\bibitem{Q2011}
Q.~{Li}, S.~S. {Choi}, Y.~{Yuan}, and D.~L. {Yao}, ``On the determination of
  battery energy storage capacity and short-term power dispatch of a wind
  farm,'' \emph{IEEE Transactions on Sustainable Energy}, vol.~2, no.~2, pp.
  148--158, Apr. 2011.

\bibitem{Ozel2014}
O.~{Ozel}, K.~{Shahzad}, and S.~{Ulukus}, ``Optimal energy allocation for
  energy harvesting transmitters with hybrid energy storage and processing
  cost,'' \emph{IEEE Transactions on Signal Processing}, vol.~62, no.~12, pp.
  3232--3245, Jun. 2014.

\bibitem{Mathias2019}
M.~Joel, M.~Robert, M.~Sean, and W.~Joseph, ``State space collapse in resource
  allocation for demand dispatch,'' in \emph{Conference on Decision and
  Control}, Sep. 2019.

\bibitem{Shi2019}
Y.~{Shi}, B.~{Xu}, Y.~{Tan}, D.~{Kirschen}, and B.~{Zhang}, ``Optimal battery
  control under cycle aging mechanisms in pay for performance settings,''
  \emph{IEEE Transactions on Automatic Control}, vol.~64, no.~6, pp.
  2324--2339, Jun. 2019.

\bibitem{Ju2018}
C.~{Ju}, P.~{Wang}, L.~{Goel}, and Y.~{Xu}, ``A two-layer energy management
  system for microgrids with hybrid energy storage considering degradation
  costs,'' \emph{IEEE Transactions on Smart Grid}, vol.~9, no.~6, pp.
  6047--6057, Nov. 2018.

\bibitem{Adefarati2013}
T.~Adefarati, ``Computational solution to economic operation of power plant,''
  \emph{Electrical and Electronic Engineering}, vol.~3, p. 139, Jan. 2013.

\bibitem{you2019role}
P.~You, D.~F. Gayme, and E.~Mallada, ``The role of strategic load participants
  in two-stage settlement electricity markets,'' \emph{arXiv:1903.08341}, 2019.

\bibitem{tang2016model}
W.~Tang, R.~Rajagopal, K.~Poolla, and P.~Varaiya, ``Model and data analysis of
  two-settlement electricity market with virtual bidding,'' in \emph{Proceding
  of IEEE Conference on Decision and Control}.\hskip 1em plus 0.5em minus
  0.4em\relax IEEE, Dec. 2016, pp. 6645--6650.

\bibitem{Liberzon2011}
D.~Liberzon, \emph{Calculus of Variations and Optimal Control Theory: A Concise
  Introduction}.\hskip 1em plus 0.5em minus 0.4em\relax Princeton, NJ, USA:
  Princeton University Press, 2011.

\bibitem{CARTIGNY20031007}
P.~{Cartigny} and P.~{Michel}, ``On a sufficient transversality condition for
  infinite horizon optimal control problems,'' \emph{Automatica}, vol.~39,
  no.~6, pp. 1007 -- 1010, Jun. 2003.

\bibitem{Kenneth1961}
K.~J. Arrow and A.~C. Enthoven, ``Quasi-concave programming,''
  \emph{Econometrica}, vol.~29, no.~4, pp. 779--800, 1961.

\bibitem{iso}
\BIBentryALTinterwordspacing
``{ISO-NE} real-time maps and charts,'' 2018. [Online]. Available:
  \url{https://www.iso-ne.com/isoexpress/}
\BIBentrySTDinterwordspacing

\bibitem{tesla}
\BIBentryALTinterwordspacing
``{Tesla Powerwall 2 Datasheet - North America}.'' [Online]. Available:
  \url{https://www.tesla.com/sites/default/files/pdfs/powerwall/Powerwall\%202_AC_Datasheet_en_northamerica.pdf}
\BIBentrySTDinterwordspacing

\end{thebibliography}

\end{document}